\documentclass[11pt]{article}
\usepackage[T1]{fontenc}
\usepackage{amsfonts}
\usepackage{amsmath}
\usepackage{amssymb}
\usepackage{amsthm}
\usepackage{bbm}
\usepackage{bm}
\usepackage{mathrsfs}
\usepackage{verbatim}
\usepackage{color}
\usepackage{pdfsync}
\usepackage{enumitem}

\newcommand{\RR}{\mathbb{R}}
\newcommand{\EE}{\mathbb{E}}
\newcommand{\PP}{\mathbb{P}}

\newcommand{\eps}{\varepsilon}
\newcommand{\1}{\mathbf{1}}

\usepackage[margin=31 mm]{geometry}
\newcommand{\mykill}[1]{}
\usepackage[pdfborder={0 0 0}]{hyperref}
\hypersetup{
  urlcolor = black,
  pdfauthor = {Marcel Nutz, Andres Riveros Valdevenito}
  pdfkeywords = {Path-dependent volatility model, SDE, wellposedness, explosion},
  pdftitle = {On the Guyon--Lekeufack Volatility Model},
  pdfsubject = {On the Guyon--Lekeufack Volatility Model},
  pdfpagemode = UseNone
}
\usepackage[capitalize]{cleveref}

\theoremstyle{plain}
\newtheorem{theorem}{Theorem}[section]
\newtheorem{proposition}[theorem]{Proposition}
\newtheorem{lemma}[theorem]{Lemma}
\newtheorem{corollary}[theorem]{Corollary}
\theoremstyle{definition}

\newtheorem{remark}[theorem]{Remark}

\theoremstyle{remark}

\begin{document}

\title{\vspace{-1em}
On the Guyon--Lekeufack Volatility Model\footnote{We thank Julien Guyon for presenting us with the question answered in this paper, and two anonymous referees for their suggestions which substantially improved the results.}}
\date{\today}
\author{
  Marcel Nutz%
  \thanks{
  Depts.\ of Statistics and Mathematics, Columbia University, mnutz@columbia.edu. Research supported by NSF Grant DMS-2106056.}
  \and
  Andr{\'e}s Riveros Valdevenito%
  \thanks{Department of Statistics, Columbia University, ar4151@columbia.edu.}
  }
  
\maketitle \vspace{-1.2em}

\begin{abstract}
Guyon and Lekeufack recently proposed a path-dependent volatility model and documented its excellent performance in fitting market data and capturing stylized facts. The instantaneous volatility is modeled as a linear combination of two processes, one is an integral of weighted past price returns and the other is the square-root of an integral of weighted past squared volatility. Each of the weightings is built using two exponential kernels reflecting long and short memory. Mathematically, the model is a coupled system of four stochastic differential equations. Our main result is the wellposedness of this system: the model has a unique strong (non-explosive) solution for all parameter values. We also study the positivity of the resulting volatility process and the martingale property of the associated exponential price process.
\end{abstract}

\vspace{.3em}

{\small

}
\vspace{.6em}

\maketitle

\section{Introduction}

Path-dependent volatility models (PDV) are stochastic models for security prices where the instantaneous volatility is a function of the price path. Starting with \cite{HobsonRogers.98,FoschiPascucci.08,Zumbach.10,ChicheporticheBouchaud.14,Guyon.14b}, such models  emphasize that prices have a feedback on volatility (e.g., the leverage effect) rather than the volatility being an exogenous factor driving the price. In their recent paper \cite{GuyonLekeufack.22}, Guyon and Lekeufack empirically study the volatility of the S\&P~500 index (and other indexes) and conclude that the majority of the variation can be explained by past index returns. Indeed, the relevant statistics are 1.\ weighted sum of past daily returns and 2.\ square-root of weighted sum of past daily squared returns (i.e., squared volatility). More specifically, long and short memory are both found to be important, hence the authors recommend using two decay kernels with different time scales. This leads to four processes feeding into the volatility: weighted sum of past returns at two timescales (indexed as (1,0) and (1,1) below) and weighted sum of past squared returns  at two (different) timescales (indexed as (2,0) and (2,1) below).

For practical purposes, \cite{GuyonLekeufack.22} finds that exponential kernels provide a tractable model with good fit. This leads the authors to propose a Markovian model with nine parameters, called the Markovian 4-factor PDV model. They convincingly argue that this model captures the important stylized facts of volatility, produces realistic price and volatility paths, and can jointly fit S\&P~500 and VIX smiles.
Specifically, the volatility process of the 4-factor PDV model is given as 
\[
  \sigma_{t} = \beta_{0} + \beta_{1}R_{1,t} + \beta_{2}\sqrt{R_{2,t}}\,.
\]
Here $R_{1,t}$ is the convex combination $(1 - \theta_{1})R_{1,0,t} + \theta_{1}R_{1,1,t}$ of the past returns weighted with different decay rates~$\lambda_{1,j}$; i.e., $R_{1,j,t}$ is an Ornstein--Uhlenbeck process $dR_{1,j,t}  = \lambda_{1,j}\sigma_{t}dW_{t} - \lambda_{1,j}R_{1,j,t}dt$ for $j\in\{0,1\}$. Moreover, $R_{2,t}$ is a convex combination $(1 - \theta_{2})R_{2,0,t} + \theta_{2}R_{2,1,t}$ of the past squared volatility weighted with different decay rates~$\lambda_{2,j}$; i.e., $R_{2,j,t}$ is an exponential moving average $dR_{2,j,t}  = \lambda_{2,j}\left(\sigma_{t}^{2} - R_{2,j,t} \right) dt$ for $j\in\{0,1\}$. 

Altogether, this leads to a coupled SDE system for the four processes $(R_{i,j,t})$, stated as~\eqref{4factor} below. Due to the square and square-root terms in the dynamics, its wellposedness is not obvious. Most importantly, it is not clear if the system explodes in finite time (the numerical simulations in \cite{GuyonLekeufack.22} truncate the volatility at a fixed upper bound). The purpose of this paper is to provide existence and uniqueness for~\eqref{4factor}. %
Our results extend to other models where the relationship between $(R_{1,t},R_{2,t})$ and $\sigma_{t}$ has a more general form satisfying certain regularity and growth properties (see \cref{2factorgen,4factorgen,2factortildegen}).

The main results are summarized in the subsequent section. There, we first discuss a simpler model which uses only one timescale for each process $R_{i,t}$, corresponding to the special case $\theta_{i}\in\{0,1\}$. While \cite{GuyonLekeufack.22} details that this 2-factor model does not provide a good fit in practice, the authors find it useful to gain intuition about the more complicated 4-factor model. Following their didactic lead, we first prove our results for the 2-factor model in \cref{se:2factor}. In this case, the equations are simpler and the algebraic expressions clearly motivate our strategy of proof. Guided by those insights, the 4-factor model can be treated using a similar strategy (detailed in \cref{se:4factor}), though the expressions are more convoluted. \Cref{se:martprop} concludes by studying the martingale property of the exponential local martingale (price) process associated with the volatility models, a  problem posed to us by an anonymous referee.

\subsection{Main Results}

The 2-factor model of \cite{GuyonLekeufack.22} is specified by an SDE driven by a standard Brownian motion~$W$,
\begin{align}
	\sigma_{t} & = \beta_{0} + \beta_{1}R_{1,t} + \beta_{2}\sqrt{R_{2,t}}  \nonumber \\
	dR_{1,t} & = \lambda_{1}\sigma_{t}dW_{t} - \lambda_{1}R_{1,t}dt  \label{2factor} \tag{\mbox{2-PDV}} \\
	dR_{2,t} & = \left(\lambda_{2}\sigma_{t}^{2} - \lambda_{2}R_{2,t} \right) dt  \nonumber
\end{align}
with parameters
\begin{align*}
  \beta_{0}, \beta_{2}, \lambda_{1},\lambda_{2} \geq 0 \quad\mbox{and}\quad \beta_{1} \leq 0
\end{align*} 
and initial values  $R_{1,0}\in \RR$ and  $R_{2,0}\in(0,\infty)$.
The above is an autonomous SDE for the processes $(R_{1,t},R_{2,t})$, with $\sigma_{t}$ merely acting as an abbreviation. On the other hand, if $\sigma_{t}$ is given, the equations for $R_{1,t}$ and $R_{2,t}$ in~\eqref{2factor} are straightforward: $R_{1,t}$ is an Ornstein--Uhlenbeck process driven by the log-returns $\sigma_{t}dW_{t}$, 
\begin{equation*}
  R_{1,t} =  R_{1,0}e^{-\lambda_{1}t} + \lambda_{1}\int_{0}^{t}e^{-\lambda_{2}(t-s)}\sigma_{s}dW_{s},
\end{equation*}
and $R_{2,t}$ is an exponential moving average of $\sigma_{t}^{2}$,
\begin{equation}\label{r2positive}
  R_{2,t} =  R_{2,0}e^{-\lambda_{2}t} + \lambda_{2}\int_{0}^{t}\sigma_{s}^{2}e^{-\lambda_{2}(t-s)}ds > 0.
\end{equation}
In particular, the expression $\sqrt{R_{2,t}}$ in~\eqref{2factor} is well-defined.%

\begin{theorem}\label{th:main2}
  The 2-factor model~\eqref{2factor} has a unique strong solution.
\end{theorem}

The proof is detailed in \cref{se:2factor}. There, we first observe that strong existence and uniqueness readily hold up to a possible explosion time, and then proceed to show the absence of explosions in finite time. \Cref{th:main2} extends to certain more general models (\cref{2factorgen}). \Cref{ta:2factorParams} reports the parameters used in~\cite{GuyonLekeufack.22}. In \cref{2sigmapos}, we provide the condition $\lambda_{2} < 2\lambda_{1}$ ensuring $\sigma_{t}>0$; that condition is satisfied by the values in \Cref{ta:2factorParams}.

\begin{table}[htb]
\centering
\begin{tabular}{|c|c|c|c|c|}
\hline
$\beta_0$ & $\beta_1$ & $\beta_2$ & $\lambda_1$ & $\lambda_2$ \\ \hline
 0.08     & {-0.08}   & 0.5       & 62          & 40         \\ \hline
\end{tabular}
\caption{Example parameters for the 2-factor model from \cite[Table~7]{GuyonLekeufack.22}}
\label{ta:2factorParams}
\end{table}

Next, we move on to the 4-factor model of \cite{GuyonLekeufack.22}. It is specified by the SDE
\begin{align}
	\sigma_{t} & = \beta_{0} + \beta_{1}R_{1,t} + \beta_{2}\sqrt{R_{2,t}} \nonumber \\
	R_{1,t} & = (1 - \theta_{1})R_{1,0,t} + \theta_{1}R_{1,1,t} \nonumber \\
	R_{2,t} & = (1 - \theta_{2})R_{2,0,t} + \theta_{2}R_{2,1,t} \label{4factor}\tag{\mbox{4-PDV}} \\
	dR_{1,j,t} & = \lambda_{1,j}\sigma_{t}dW_{t} - \lambda_{1,j}R_{1,j,t}dt, \quad j \in \{0,1\} \nonumber \\
	dR_{2,j,t} & = \lambda_{2,j}\left(\sigma_{t}^{2} - R_{2,j,t} \right) dt, \quad j \in \{0,1\} \nonumber 
\end{align}
with parameters
\begin{align*}
  \beta_{0}, \beta_{2}, \lambda_{1,j},\lambda_{2,j} \geq 0, \qquad \beta_{1} \leq 0, \qquad \theta_{1}, \theta_{2} \in [0,1]
\end{align*} 
and initial values  $R_{1,j,0} \in \RR$ and  $R_{2,j,0} > 0$, $j \in \{0,1\}$.
The above is an autonomous SDE for the four processes $(R_{1,j,t},R_{2,j,t})_{j \in \{0,1\}}$. We note that~\eqref{4factor} generalizes the 2-factor model~\eqref{2factor}; the latter is recovered when $\theta_{1}, \theta_{2} \in \{0,1\}$. Whereas for $\theta_{1}, \theta_{2} \in (0,1)$, the difference with~\eqref{2factor} is that $R_{1,t},R_{2,t}$ are proper convex combinations of processes with different time scales. Once again, $R_{1,j,t}$ and $R_{2,j,t}$ have straightforward expressions once $\sigma_{t}$ is given, and $R_{2,j,t}>0$ as in \eqref{r2positive}. Our main result reads as follows.

\begin{theorem}\label{th:main4}
   The 4-factor model~\eqref{4factor} has a unique strong solution.
\end{theorem}

The proof is stated in \cref{se:4factor}. Again, strong existence and uniqueness readily hold up to a possible explosion time, and we prove absence of explosions in finite time. \Cref{th:main4} extends to certain more general models (\cref{4factorgen}). The parameters used in~\cite{GuyonLekeufack.22} are reproduced in \Cref{ta:4factorParams}. 
While in the 2-factor model, $\sigma_t$ remains strictly positive for a certain parameter range (\cref{2sigmapos}), that property can fail in the 4-factor model (\cref{pr:4sigmaposFail}).

\begin{table}[htb]
\centering
\begin{tabular}{|c|c|c|c|c|c|c|c|c|}
\hline
$\beta_0$ & $\beta_1$ & $\beta_2$ & $\lambda_{1,0}$ & $\lambda_{1,1}$ & $\lambda_{2,0}$ & $\lambda_{2,1}$ & $\theta_{1}$ & $\theta_{2}$\\ \hline
 0.04     & {-0.13}   & 0.65      & 55              & 10              &      20         &       3         & 0.25         &   0.5   \\ \hline
\end{tabular}
\caption{Example parameters for the 4-factor model from \cite[Table~8]{GuyonLekeufack.22}}
\label{ta:4factorParams}
\end{table}

Finally, we study the martingale property of the resulting price process $(X_{t})_{t \geq 0}$, a  problem posed by an anonymous referee. We provide a positive result for the 2-factor model; the problem remains open for the 4-factor model.
For the sake of generality, we allow processes $(\sigma_{t})_{t \geq 0}$ that can become negative (even if those may be undesirable in practice), but stop them at some level for technical reasons; that gives rise to the volatility process $(\nu_{t})_{t \geq 0}$ in the theorem below. If the process $(\sigma_{t})_{t \geq 0}$ is nonnegative, as is guaranteed for the parameters mentioned in~\cref{2sigmapos}, then clearly $\nu_{t}=\sigma_{t}$.
  
\begin{theorem}\label{martprop}
  Let $(\sigma_{t})_{t\geq0}$ be given by~\eqref{2factor} and $\nu_{t}:=\sigma_{t\wedge \tau}$ where $\tau = \inf \left\{ t \geq 0:\sigma_{t} < -C \right\}$ for some $C\geq0$. The exponential local martingale $(X_{t})_{t\geq 0}$ given by 
\begin{align*}
    dX_{t} = \nu_{t}X_{t}dW_{t}, \quad X_{0} = x_{0} > 0
\end{align*}
  is a true martingale.
\end{theorem}

The proof is reported in \cref{se:martprop}. Following an idea in~\cite{Sin.98} and~\cite{JarrowProtterSanMartin.22}, we characterize the martingale property of $(X_{t})_{t \geq 0}$ as the non-explosiveness of $(\nu_{t})_{t \geq 0}$ under a changed measure, and then prove the latter by an estimate for the associated stochastic differential equation. This line of argument may extend to the 4-factor model, but the present argument for non-explosiveness in the 2-factor case does not apply in the 4-factor case (see \cref{rk:tiltedPDV4}).

\section{Analysis of the 2-Factor Model~\eqref{2factor}}\label{se:2factor}

We first show, using fairly standard arguments, that~\eqref{2factor} has a unique strong solution  up to a possible explosion time. Then, we prove the absence of explosions. In \cref{se:2factorpos}, we study the positivity of $\sigma_{t}$.

\subsection{Wellposedness and Absence of Explosions}\label{se:2factorWellposed}

To detail the aforementioned arguments, we introduce a more concise notation for~\eqref{2factor}: writing $R_{t} := \left( R_{1,t}, R_{2,t} \right)$, we can rewrite~\eqref{2factor} as 
\begin{gather*}
	dR_{t}  = b(R_{t})dt + \nu(R_{t})dW_{t}, \qquad R_{0}  = (R_{1,0}, R_{2,0}) \\[.3em]  %
	\nu(x,y) = \begin{pmatrix} \lambda_{1}(\beta_{0} + \beta_{1}x + \beta_{2}\sqrt{y}) \\
	0 	
	\end{pmatrix} \\[.3em]
	b(x,y) = \begin{pmatrix} -\lambda_{1}x \\
		\lambda_{2} \left(\beta^{2}_{0} + \beta^{2}_{1}x^{2} + (\beta^{2}_{2} - 1)y  + 2\beta_{0}\beta_{1}x + 2\beta_{0}\beta_{2}\sqrt{y} + 2\beta_{1}\beta_{2}x\sqrt{y}\right)
	\end{pmatrix}.
\end{gather*}
As the coefficients $\nu(x,y)$ and $b(x,y)$ are continuous in their domains and the initial condition is deterministic, the general existence result of \cite[Theorem IV.2.3]{IkedaWatanabe.89} (applied with $\sqrt{y}:=0$ for $y<0$) shows the following.

\begin{lemma} \label{weakexist} The SDE \eqref{2factor} has a weak solution up to a possible explosion time.
\end{lemma}

Next, we establish pathwise uniqueness. The usual local Lipschitz condition (e.g., \cite[Theorem~IV.3.1]{IkedaWatanabe.89}) fails because of the term $\sqrt{y}$ in the coefficients. However, as this failure only occurs at the boundary of the relevant domain, a modification of the usual proof applies. 

\begin{lemma} \label{pathwiseUnique} 
The SDE \eqref{2factor} satisfies pathwise uniqueness.
\end{lemma}

\begin{proof}
Following the proof of \cite[Theorem~IV.3.1]{IkedaWatanabe.89}, we consider two solutions $(R, W)$ and $(R', W)$ of~\eqref{2factor} on the same probability space $\left( \Omega, \mathcal{F}, \PP \right)$ with Brownian motion~$W$, and with the same initial values $R_{0} = R'_{0}= (R_{1,0}, R_{2,0})$.

Given $N,\eps>0$, there exists $K_{\eps,N} > 0$ such that 
\begin{align*}
	\left\| \nu(x,y) - \nu(x',y') \right\|^{2} + \left\| b(x,y) - b(x',y') \right\|^{2} \leq K_{\eps,N} \left\| (x,y) - (x',y') \right\|^2
\end{align*}
for all $(x,y),(x',y')\in\RR^{2}$ with $\|(x,y)\|,\|(x',y')\|\leq N$ and $y,y' \geq \eps $. 
Define
\begin{align*}
	S_{N}  := \inf \left\{ t \geq 0 \text{  :  } \left\| R_{t} \right\| \geq N \right\}, \qquad T_{\eps}  := \inf \left\{ t \geq 0 \text{  :  } \left| R_{2,t} \right| \leq \eps \right\}
\end{align*}
and similarly $S_{N}', T_{\eps}'$ for $R'$ instead of $R$.
Set $S_{\eps, N} := S_{N} \wedge S'_{N} \wedge T_{\eps} \wedge T'_{\eps}$ and note that
\begin{align*}
	R_{t \wedge S_{\eps, N}} - R'_{t \wedge S_{\eps, N}} = \int_{0}^{t \wedge S_{\eps, N}} \left(\nu(R_{s} ) - \nu(R'_{s} )\right) dW_{s} + \int_{0}^{t \wedge S_{\eps, N}} \left(b(R_{s} ) - b(R'_{s} )\right) ds.
\end{align*}
Fix $T\in (0,\infty)$. For $t \leq T$, It\^o's isometry and H\"older's inequality yield
\begin{align*}
		&\EE \left(\|R_{t \wedge S_{\eps, N}}  - R'_{t \wedge S_{\eps, N}}\|^{2} \right) \\ 
		   & \leq 2 \EE \left( \left\|\int_{0}^{t \wedge S_{\eps, N}} \left(\nu(R_{s} ) - \nu(R'_{s} )\right) dW_{s} \right\|^{2} \right) + 2 \EE \left( \left\|\int_{0}^{t \wedge S_{\eps, N}} \left(b(R_{s} ) - b(R'_{s} )\right) ds \right\|^{2} \right) \\
		& \leq 2 \EE \left( \int_{0}^{t} \left\|\nu(R_{s \wedge S_{\eps, N}} ) - \nu(R'_{s \wedge S_{\eps, N}} ) \right\|^{2} ds \right) + 2T \EE \left( \int_{0}^{t} \left\|b(R_{s \wedge S_{\eps, N}} ) - b(R'_{s \wedge S_{\eps, N}} ) \right\|^{2}ds  \right) \\
		& \leq 2(1 + T)K_{\eps, N}\EE \left( \int_{0}^{t} \left\|R_{s \wedge S_{\eps, N}}  - R'_{s \wedge S_{\eps, N}} \right\|^{2} ds \right),
\end{align*}
and then Gr\"onwall's inequality shows
\begin{align*}
	\EE \left(\|R_{t \wedge S_{\eps, N}} - R'_{t \wedge S_{\eps, N}}\|^{2} \right) = 0.
\end{align*}
As $T > 0$ was arbitrary, this holds for all $t\geq0$. 

Next, we let $\eps\to0$. In view of the positivity~\eqref{r2positive}, we have $T_{\eps}, T'_{\eps} \to \infty$  and conclude that
\begin{align*}
	R_{t \wedge  S_{N} \wedge S'_{N}} = R'_{t \wedge  S_{N} \wedge S'_{N}} \quad  \forall t \geq 0.
\end{align*}
Together with the continuity of the paths, it follows that $S_{N} = S'_{N}$, and since this holds for all $N>0$, we have shown that $R=R'$ up to a possible (common) time of explosion.
\end{proof}

As weak existence together with pathwise uniqueness implies strong existence \cite[Theorem~IV.1.1]{IkedaWatanabe.89}, the results so far establish the strong wellposedness of~\eqref{2factor} up to explosion.

\begin{corollary} \label{co:wellposed2} The SDE \eqref{2factor} satisfies strong existence and  uniqueness up to a possible explosion time.
\end{corollary} 

Turning to the main contribution of this section, we now show the absence of explosions.

\begin{lemma} \label{2factornonexplosion} %
A solution $(R_{1,t},R_{2,t})$ of~\eqref{2factor} cannot explode in finite time. Moreover, $\sup_{t\leq T}\EE(R^{2}_{1,t} + R_{2,t})<\infty$ for any $T\in[0,\infty)$.
\end{lemma}

\begin{proof}
  Fix $M > 0$ and define the stopping times
	\begin{gather*}
	T^{1}_{M}  := \inf \left\{ t \geq 0 \text{ : } R^{2}_{1,t} \geq M^{2} \right\} \qquad
	T^{2}_{M}  := \inf \left\{ t \geq 0 \text{ : } R_{2,t} \geq M^{2} \right\} \\
	T_{M}  := T^{1}_{M} \wedge T^{2}_{M}.
\end{gather*}
Fix also $t \geq 0$, and note that
\begin{align*}
	M^{2} \PP(T_{M} \leq t) & = \EE(M^{2} \1_{T_{M} \leq t}) 
	 \leq \EE( \max(R^{2}_{1,T_{M} \wedge t},R_{2,T_{M}\wedge t} )) 
	 \leq \EE(R^{2}_{1,T_{M} \wedge t} + R_{2,T_{M} \wedge t}).
\end{align*}
In the main part of the proof below, we show that
\begin{equation}\label{eq:2factornonexplosionProof}
  \EE(R^{2}_{1,T_{M} \wedge t} + R_{2,T_{M} \wedge t})\leq c(t)
\end{equation}
with $c(t)<\infty$ independent of~$M$. It will then follow that $\lim_{M\to\infty}\PP(T_{M} \leq t)=0$, showing that $R_{1}$ and $R_{2}$ have bounded paths on any compact time interval and hence completing the proof.

To show~\eqref{eq:2factornonexplosionProof}, we first apply It\^o's formula to obtain
\begin{align*}
	R^{2}_{1, t \wedge T_{M}} 
	& = R^{2}_{1,0} + \int_{0}^{t \wedge T_{M}}2\lambda_{1}\sigma_{s}R_{1,s}dW_{s} + \int_{0}^{t \wedge T_{M}} (\lambda^{2}_{1} \sigma^{2}_{s} - 2\lambda_{1}R^{2}_{1,s}) ds.
\end{align*}
As $\sigma_{s} R_{1,s}$ is uniformly bounded up to the stopping time $t \wedge T_{M}$, it follows that
\begin{align*}
	\EE(R^{2}_{1,T_{M} \wedge t}) = R_{1,0}^{2} + \EE\left( \int_{0}^{t \wedge T_{M}} (\lambda^{2}_{1} \sigma^{2}_{s} - 2\lambda_{1}R^{2}_{1,s}) ds \right)
\end{align*}
and thus, by Fubini's theorem,
\begin{align}\label{eq:2factorItoExpBound}
	\EE(R^{2}_{1,T_{M} \wedge t}) %
	& = R_{1,0}^{2} +  \int_{0}^{t}\EE\left( (\lambda^{2}_{1} \sigma^{2}_{s} - 2\lambda_{1}R^{2}_{1,s}) \1_{s \leq t \wedge T_{M}}\right) ds.
\end{align}
Next, we insert the definition of $\sigma^{2}_{s}$ to get
\begin{align*}
	 \EE(R^{2}_{1,T_{M} \wedge t})  
	= R_{1,0}^{2}
	+ \int_{0}^{t}\EE\, \Big( &\Big\{\lambda^{2}_{1}(\beta_{0} + \beta_{2}\sqrt{R_{2,s}})^{2} \\
	&  + 2\lambda^{2}_{1}\beta_{1}R_{1,s}(\beta_{0} + \beta_{2}\sqrt{R_{2,s}}) + (\lambda^{2}_{1}\beta_{1}^{2} - 2\lambda_{1})R^{2}_{1,s}\Big\} \1_{s \leq t \wedge T_{M}}\Big)\, ds.
\end{align*} 
Using the elementary inequalities 
\begin{align*}
    2ab \leq a^{2} + b^{2} \qquad \mbox{and} \qquad 
    (a + b)^{2} \leq 2a^{2} + 2b^{2}
\end{align*}
we deduce
\begin{align*}
    \EE(R^{2}_{1,T_{M} \wedge t}) \leq R^{2}_{1,0} + \int_{0}^{t}\EE \Big( & \Big\{ 2\lambda_{1}^{2}\beta^{2}_{0} + 2\lambda^{2}_{1}\beta_{2}^{2}R_{2,s} + \lambda_{1}^{2}\beta^{2}_{0} + \lambda_{1}^{2}\beta^{2}_{1}R_{1,s}^{2}  \\
    & \;\;+ \lambda_{1}^{2}\beta_{1}^{2}R_{1,s}^{2} + \lambda_{1}^{2}\beta_{2}^{2}R_{2,s} + (\lambda^{2}_{1}\beta_{1}^{2} - 2\lambda_{1})R^{2}_{1,s} \Big\} \1_{s \leq t \wedge T_{M}}\Big) ds 
\end{align*}
and therefore
\begin{align}
    \EE(R^{2}_{1,T_{M} \wedge t}) & \leq R^{2}_{1,0} + \int_{0}^{t}\EE\left( \left(3\lambda_{1}^{2}\beta^{2}_{0} + 3\lambda_{1}^{2}\beta_{2}^{2}R_{2,s} + (3\lambda_{1}^{2}\beta_{1}^{2} - 2\lambda_{1})R_{1,s}^{2} \right) \1_{s \leq t \wedge T_{M}}\right) ds \nonumber \\
    & \leq R_{1,0}^{2} + 3\lambda_{1}^{2}\beta^{2}_{0}t + \max\left\{3\lambda_{1}^{2}\beta_{2}^{2}, (3\lambda_{1}^{2}\beta_{1}^{2} - 2\lambda_{1}) \right\} \int_{0}^{t}\EE\left( R_{1,s \wedge T_{M}}^{2} + R_{2,s \wedge T_{M}}\right) ds \nonumber \\
     & =: c_{1,1} + c_{1,2}t  + c_{1,3}\int_{0}^{t}\EE\left(  R_{1,s \wedge T_{M}}^{2} + R_{2,s \wedge T_{M}}\right) ds \label{2factorr1eq}
\end{align}
where the constants $c_{1,1},c_{1,2},c_{1,3}>0$ are independent of $M$ and $t$.

Our next goal is a similar bound for $R_{2}$ instead of $R^{2}_{1}$. From the SDE for $R_{2}$,
\begin{align*}
	&\EE(R_{2,T_{M} \wedge t}) \\
	& = R_{2,0} +\lambda_{2}\EE\left( \int_{0}^{t \wedge T_{M}} \left(\sigma^{2}_{s} - R_{2,s} \right)ds  \right) \\
	& = R_{2,0} +\lambda_{2}\int_{0}^{t} \EE\left( (\sigma^{2}_{s} - R_{2,s}) \1_{s \leq t \wedge T_{M}}\right) ds  \\
	& = R_{2,0} + \lambda_{2} \int_{0}^{t}\EE\left( \left\{\left(\beta_{0} + \beta_{1}R_{1,s}\right)^{2} + 2\left(\beta_{0} + \beta_{1}R_{1,s}\right)\beta_{2}\sqrt{R_{2,s}} + (\beta^{2}_{2} - 1)R_{2,s}\right\} \1_{s \leq t \wedge T_{M}}\right)ds.
\end{align*}
Similarly as above, we obtain
\begin{align*}
    \EE(R_{2,T_{M} \wedge t}) \leq R_{2,0} + \lambda_{2} \int_{0}^{t}\EE \Big( & \Big\{ 2\beta^{2}_{0} + 2\beta_{1}^{2}R^{2}_{1,s} + \beta^{2}_{0} + \beta_{2}^{2}R_{2,s} \\
    & \;\;+ \beta_{1}^{2}R_{1,s}^{2} + \beta_{2}^{2}R_{2,s} + (\beta_{2}^{2} - 1)R_{2,s} \Big\} \1_{s \leq t \wedge T_{M}}\Big) ds .
\end{align*}
We conclude that
\begin{align}
	\EE(R_{2,T_{M} \wedge t}) & \leq R_{2,0} + \lambda_{2} \int_{0}^{t}\EE\left( \left( 3\beta^{2}_{0} + 3\beta^{2}_{1}R^{2}_{1,s} + (3\beta_{2}^{2} - 1)R_{2,s} \right) \1_{s \leq t \wedge T_{M}}\right) ds \nonumber \\
    & \leq  R_{2,0} +  3\lambda_{2}\beta^{2}_{0}t + 
    \max\{3\beta^{2}_{1}, 3\beta^{2}_{2} - 1 \} \int_{0}^{t}\EE\left( R_{1,s \wedge T_{M}}^{2} +  R_{2,s \wedge T_{M}}\right) ds \nonumber \\
    & =: c_{2,1} + c_{2,2} t  + c_{2,3}\int_{0}^{t}\EE\left( R_{1,s \wedge T_{M}}^{2} +  R_{2,s \wedge T_{M}}\right) ds, \label{2factorr2eq}
\end{align}
where again the constants do not depend on $M$ and $t$.

Writing $c_{i}=c_{1,i}+c_{2,i}$, 
combining~\eqref{2factorr1eq} and~\eqref{2factorr2eq} yields
\begin{align*}
	\EE(	R^{2}_{1,T_{M} \wedge t} + R_{2,T_{M} \wedge t}) \leq c_{1} + c_{2}t + c_{3}\int_{0}^{t}\EE\left( R_{1,s \wedge T_{M}}^{2} +   R_{2,s \wedge T_{M}}\right) ds,
\end{align*}
and now Gr\"onwall's inequality shows
\begin{align*}
	\EE( R^{2}_{1,T_{M} \wedge t} + R_{2,T_{M} \wedge t}) \leq \left(c_{1} + c_{2}t\right)e^{c_{3}t}.
\end{align*}
This establishes~\eqref{eq:2factornonexplosionProof} and hence completes the proof.
\end{proof}

\begin{remark}\label{2factorgen}
    The results in this section generalize to volatility models having the same dynamics as~\eqref{2factor} for $(R_{1,t},R_{2,t})$ but a more general functional $\sigma_{t} = f(R_{1,t},R_{2,t})$ where
  \begin{enumerate}
  \item $f : \RR \times \RR_{+} \to \RR$ is continuous, and Lipschitz on compact subsets of $\RR \times (0,\infty)$,
  \item there are $K_{1}, K_{2} \in \RR$ such that
  \begin{align}
            f(x,y)^{2} \leq K_{1}(x^{2} + y) + K_{2}. \label{2factorgenbound}
  \end{align}   
  \end{enumerate} 
Indeed, continuity and local Lipschitz continuity imply analogues of \cref{weakexist} and \cref{pathwiseUnique}. Inequality \eqref{2factorgenbound} implies that bounding~\eqref{eq:2factornonexplosionProof} is enough to bound the expectation of $\sigma^{2}_{t}$. Moreover, it implies that 
        \begin{align*}
            \lambda_{1}^{2}f(x,y)^{2} - 2\lambda_{1}x^{2} \leq K_{1,1}(x^{2} + y) + K_{1,2} %
        \end{align*}
        for some $K_{1,1}, K_{1,2} \in \RR$, so that a bound analogous to~\eqref{2factorr1eq} holds, as well as
        \begin{align*}
            f(x,y)^{2} - y \leq K_{1}(x^{2} + y) + K_{2}, %
        \end{align*}
        so that a bound analogous to~\eqref{2factorr2eq} holds.
\end{remark}

\subsection{Positivity of $\sigma_{t}$}\label{se:2factorpos}

In this section we show that $\sigma_{t}>0$ under certain conditions. More precisely, we exhibit a lower bound $\sigma_{t}\geq Y_{t}>0$ which also shows, e.g., that $1/\sigma_{t}$ has finite moments of all orders. This strengthens the result in \cite[Section~4.1.5]{GuyonLekeufack.22} where it is observed that $\sigma_{t}\geq0$ since the drift of $\sigma_{t}$ would be positive whenever $\sigma_{t}$ reaches~$0$.

\begin{theorem}\label{2sigmapos} 
  Consider the solution $(R_{1,t},R_{2,t})$ of~\eqref{2factor} up to a possible time $\tau$ of explosion. If the initial values $(R_{1,0}, R_{2,0})$ are such that $\sigma_{0} = \beta_{0} + \beta_{1}R_{1,0} + \beta_{2} \sqrt{R_{2, 0}}>0$, and if moreover $\lambda_{2} < 2\lambda_{1}$, then $\sigma_{t} > 0$ for all $t < \tau$. More precisely, we have $\sigma_{t} \geq Y_{t}$, where $Y_{t}$ is the stochastic exponential~\eqref{eq:Y}.
\end{theorem}

\begin{proof}
By It\^o's formula, $\sigma_{t}$ satisfies the SDE
\begin{align}
	d\sigma_{t} = \left(-\beta_{1}\lambda_{1}R_{1,t} + \frac{\lambda_{2}\beta_{2}}{2} \frac{\sigma^{2}_{t} - R_{2,t}}{\sqrt{R_{2,t}}} \right) dt + \beta_{1}\lambda_{1}\sigma_{t}dW_{t} \,.\label{2factorsigmaeq}
\end{align}
Using the assumption that $\lambda_{2} < 2\lambda_{1}$, we can bound the drift of $\sigma_{t}$ from below:
\begin{align}
	-\beta_{1}\lambda_{1}R_{1,t} + \frac{\lambda_{2}\beta_{2}}{2} \frac{\sigma^{2}_{t} - R_{2,t}}{\sqrt{R_{2,t}}} & = -\lambda_{1}(\sigma_{t} - \beta_{0} - \beta_{2} \sqrt{R_{2,t}}) - \frac{\lambda_{2}\beta_{2}}{2} \sqrt{R_{2, t}} + \frac{\lambda_{2}\beta_{2}}{2}\frac{\sigma^{2}_{t}}{\sqrt{R_{2,t}}} \nonumber \\
	& \geq -\lambda_{1}\sigma_{t} + \beta_{0}\lambda_{1} + \beta_{2} \left(\lambda_{1} - \frac{\lambda_{2}}{2} \right) \sqrt{R_{2, t}} + \frac{\lambda_{2}\beta_{2}}{2}\frac{\sigma^{2}_{t}}{\sqrt{R_{2,t}}} \nonumber \\
	& \geq -\lambda_{1}\sigma_{t} \label{driftineq}
\end{align}
because all the other terms are nonnegative. Inspired by~\eqref{driftineq}, we define a process $Y$ via
\begin{align}
	dY_{t} = -\lambda_{1}Y_{t}dt + \beta_{1}\lambda_{1}Y_{t}dW_{t} \text{ , } \quad Y_{0} = \sigma_{0}. \label{compsigma}
\end{align}
Note that $Y$ is simply the stochastic exponential
\begin{align}\label{eq:Y}
	Y_{t} = \sigma_{0}\exp \left( \beta_{1}\lambda_{1}W_{t} - \lambda_{1}t - \frac{1}{2} \beta^{2}_{1}\lambda^{2}_{1} t \right)
\end{align}
and, in particular, $Y_{t} > 0$ for all $t \geq 0$.

The SDEs \eqref{2factorsigmaeq} and \eqref{compsigma} have the same initial condition $\sigma_{0}$, the same volatility function $v(x) = \beta_{1}\lambda_{1}x$ and their drift functions are ordered according to \eqref{driftineq}. Moreover, both drift and volatility functions are continuous on the relevant domains, the drift function of $Y$ is Lipschitz, and
\begin{align*}
	\int_{0}^{\eps} v(x)^{-2} dx = \infty
\end{align*}
for every $\eps > 0$. In view of these conditions, the comparison result for SDEs\footnote{To be precise, the cited theorem is stated for SDEs where the drift and volatility functions depend only on time and the solution process. Here, they are random as they depend on $(R_{1,t},R_{2,t})$. The proof holds without change.} \cite[Theorem~5.2.18]{KaratzasShreve.91}  yields that  $\sigma_{t} \geq Y_{t}$ for all $t<\tau$. 
\end{proof}

\section{Analysis of the 4-Factor Model~\eqref{4factor}}\label{se:4factor}

In \cref{se:4factorWellposed} we prove our main result on the wellposedness of the 4-factor model and the absence of explosions. In \cref{se:4factorpos} we show that $\sigma_t$ need not remain positive under the stated conditions.

\subsection{Wellposedness and Absence of Explosions}\label{se:4factorWellposed}

The general wellposedness of~\eqref{4factor} is shown using the same arguments as in the 2-factor model; we therefore omit the proof.

\begin{proposition} \label{pr:wellposed2} The SDE \eqref{4factor} satisfies strong existence and uniqueness up to a possible explosion time.
\end{proposition} 

The next result contains our main contribution.

\begin{lemma} \label{4factornonexplosion} A solution $(R_{1,j,t},R_{2,j,t})_{j\in\{0,1\}}$ of~\eqref{4factor} cannot explode in finite time. Moreover, $\sup_{t\leq T}\EE(R^{2}_{1,0, t} + R^{2}_{1,1, t } + R_{2,0, t } + R_{2,1, t})<\infty$ for any $T\in[0,\infty)$.
 \end{lemma}

\begin{proof}
  We follow the guidance provided by the 2-factor model. Fix $M>0$ and define
\begin{align*}
	T^{1,0}_{M} := \inf \left\{ t \geq 0 \text{ : } R^{2}_{1,0,t} \right. & \left. \geq M^{2} \right\} \qquad T^{1,1}_{M} := \inf \left\{ t \geq 0 \text{ : } R^{2}_{1,1,t} \geq M^{2} \right\} \\
	T^{2,0}_{M} := \inf \left\{ t \geq 0 \text{ : } R_{2,0,t}  \right. & \left. \geq M^{2} \right\} \qquad T^{2,1}_{M} := \inf \left\{ t \geq 0 \text{ : } R_{2,1,t} \geq M^{2} \right\} \\
	T_{M} & := T^{1,0}_{M} \wedge  T^{1,1}_{M} \wedge  T^{2,0}_{M} \wedge  T^{2,1}_{M}.
\end{align*}
Fix also $t>0$, then
\begin{align*}
	M^{2}\PP(T_{M} \leq t) & = \EE(M^{2}\1_{T_{M} \leq t}) \\
	& \leq \EE\left( \max\left(R^{2}_{1,0,t \wedge T_{M}}, R^{2}_{1,1,t \wedge T_{M}}, R_{2,0,t \wedge T_{M}}, R_{2,1,t \wedge T_{M}} \right) \right) \\
	& \leq \EE\left( R^{2}_{1,0,t \wedge T_{M}} + R^{2}_{1,1,t \wedge T_{M}} + R_{2,0,t \wedge T_{M}} + R_{2,1,t \wedge T_{M}} \right) = \EE(U_{t \wedge T_{M}})
\end{align*}
where
\begin{align*}
	U_{t} := R^{2}_{1,0, t} + R^{2}_{1,1, t } + R_{2,0, t } + R_{2,1, t} \,.
\end{align*}
We shall prove that
\begin{equation}\label{eq:4factornonexplosionProof}
  \EE(U_{t \wedge T_{M}})\leq C(t)
\end{equation}
for a continuous $C(t)$ independent of~$M$, and that will imply the claim.

To prove~\eqref{eq:4factornonexplosionProof}, we note as in~\eqref{eq:2factorItoExpBound} that 
\begin{align}
	\EE\left( R^{2}_{1,j, t \wedge T_{M}} \right) 
	& = R^{2}_{1,j,0} + \int_{0}^{t}  \EE \left(\{\lambda^{2}_{1,j} \sigma^{2}_{s} - 2\lambda_{1,j}R^{2}_{1,j,s}\} \1_{s \leq t \wedge T_{M}} \right) ds. \label{4factorr1eq} 
\end{align}
We first focus on $j=0$. Inserting the definition of $\sigma^{2}_{s}$ and then the one of $R^{2}_{1,s}$, we have 
\begin{align*}
	\lambda^{2}_{1,0} \sigma^{2}_{s} & -  2\lambda_{1,0}R^{2}_{1,0,s} \\
    &= \lambda^{2}_{1,0}(\beta_{0} + \beta_{2}\sqrt{R_{2,s}})^{2} + 2\lambda^{2}_{1,0}\beta_{1}(\beta_{0} + \beta_{2}\sqrt{R_{2,s}})R_{1,s} + \lambda^{2}_{1,0}\beta^{2}_{1}R^{2}_{1,s}  - 2\lambda_{1,0}R^{2}_{1,0,s} \\ 
    &\leq  3\lambda_{1,0}^{2}\beta_{0}^{2} + 3\lambda_{1,0}^{2}\beta_{2}^{2}R_{2,s} + 3\lambda_{1,0}^{2}\beta_{1}^{2}R^{2}_{1,s} - 2\lambda_{1,0}R^{2}_{1,0,s} \\
    & \leq 3\lambda_{1,0}^{2}\beta_{0}^{2} + 3\lambda_{1,0}^{2}\beta_{2}^{2}(R_{2,0,s} + R_{2,1,s}) + 3\lambda_{1,0}^{2}\beta_{1}^{2}(R^{2}_{1,0,s} + R^{2}_{1,1,s}) - 2\lambda_{1,0}R^{2}_{1,0,s} \\
    & = 3\lambda_{1,0}^{2}\beta_{0}^{2} + 3\lambda_{1,0}^{2}\beta_{2}^{2}R_{2,0,s} + 3\lambda_{1,0}^{2}\beta_{2}^{2}R_{2,1,s}+ 3\lambda_{1,0}^{2}\beta_{1}^{2}R^{2}_{1,1,s} + (3\lambda_{1,0}^{2}\beta_{1}^{2} - 2\lambda_{1,0})R^{2}_{1,0,s}
\end{align*}
where we used the elementary convexity inequalities
\begin{align}
	R^{2}_{1, s} & =((1 - \theta_{1})R_ {1,0,s} + \theta_{1}R_{1,1,s})^{2} \leq (1 - \theta_{1})R^{2}_ {1,0,s} + \theta_{1}R^{2}_{1,1,s} \leq R^{2}_ {1,0,s} + R^{2}_{1,1,s} \,,\label{r1convexineq} \\
 R_{2,s} & = (1 - \theta_{2})R_ {2,0,s} + \theta_{2}R_{2,1,s} \leq R_ {2,0,s} + R_{2,1,s} \,. \label{r2convexineq}
\end{align}
Using this inequality in~\eqref{4factorr1eq} yields
\begin{align*}
	\EE\left( R^{2}_{1,0, t \wedge T_{M}} \right) 
	& \leq R^{2}_{1,0,0} + \int_{0}^{t}  \EE  \Big( 3\lambda_{1,0}^{2}\beta_{0}^{2} + 3\lambda_{1,0}^{2}\beta_{2}^{2}R_{2,0,s \wedge T_{M}}\\
 & \quad\, + 3\lambda_{1,0}^{2}\beta_{2}^{2}R_{2,1,s\wedge T_{M}}+ 3\lambda_{1,0}^{2}\beta_{1}^{2}R^{2}_{1,1,s\wedge T_{M}} + (3\lambda_{1,0}^{2}\beta_{1}^{2} - 2\lambda_{1,0})R^{2}_{1,0,s\wedge T_{M}} \Big) ds \\
	& =: c_{1,0,0}(t) + c_{1,0,1}\int_{0}^{t} \EE(R_{2,0, s \wedge T_{M}})ds + c_{1,0,1}\int_{0}^{t} \EE(R_{2,1, s \wedge T_{M}})ds \\
    & \quad\, + c_{1,0,1} \int_{0}^{t} \EE(R^{2}_{1,1, s \wedge T_{M}})ds + c_{1,0,2} \int_{0}^{t} \EE(R^{2}_{1,0, s \wedge T_{M}})ds
\end{align*} 
where $c_{1,0,0}(t)$ is affine in~$t$.
Symmetrically, we have for $j = 1$ that
\begin{align*}
	\lambda^{2}_{1,1} \sigma^{2}_{s} - 2\lambda_{1,1}R^{2}_{1,1,s} 
	& \leq 3\lambda_{1,1}^{2}\beta_{0}^{2} + 3\lambda_{1,1}^{2}\beta_{2}^{2}R_{2,0,s} \\
    & \quad\, + 3\lambda_{1,1}^{2}\beta_{2}^{2}R_{2,1,s}+ 3\lambda_{1,1}^{2}\beta_{1}^{2}R^{2}_{1,0,s} + (3\lambda_{1,1}^{2}\beta_{1}^{2} - 2\lambda_{1,1})R^{2}_{1,1,s}
\end{align*}
and then
\begin{align*}
	\EE\left( R^{2}_{1,1, t \wedge T_{M}} \right) & \leq R^{2}_{1,1,0} + \int_{0}^{t}  \EE  \Big( 3\lambda_{1,1}^{2}\beta_{0}^{2} + 3\lambda_{1,1}^{2}\beta_{2}^{2}R_{2,0,s \wedge T_{M}}  \\
 & \quad\, + 3\lambda_{1,1}^{2}\beta_{2}^{2}R_{2,1,s\wedge T_{M}}+ 3\lambda_{1,1}^{2}\beta_{1}^{2}R^{2}_{1,0,s\wedge T_{M}} + (3\lambda_{1,1}^{2}\beta_{1}^{2} - 2\lambda_{1,1})R^{2}_{1,1,s\wedge T_{M}} \Big) ds \\
	& =: c_{1,1,0}(t) + c_{1,1,1}\int_{0}^{t} \EE(R_{2,0,s \wedge T_{M}})ds + c_{1,1,1}\int_{0}^{t} \EE(R_{2,1,s \wedge T_{M}})ds \\
    & \quad\, + c_{1,1,1} \int_{0}^{t} \EE(R^{2}_{1,0, s \wedge T_{M}})ds + c_{1,1,2} \int_{0}^{t} \EE(R^{2}_{1,1, s \wedge T_{M}})ds.
\end{align*}

Combining the results for $j=0$ and $j=1$ yields
\begin{align}
	\EE\left(  R^{2}_{1,0, t \wedge T_{M}} \right. & +\left. R^{2}_{1,1, t \wedge T_{M}} \right) \nonumber \\
	&\leq c_{1,0}(t) + c_{1,1}\int_{0}^{t} \EE(R_{2,0,s \wedge T_{M}} + R_{2,1,s \wedge T_{M}})ds  \nonumber \\
 & \quad\, + c_{1,2} \int_{0}^{t} \EE\left(  R^{2}_{1,0, s \wedge T_{M}} + R^{2}_{1,1, s \wedge T_{M}} \right)ds \label{4factorr1sumeq}
\end{align}
where the constants have the obvious definitions.

Next, we derive a similar bound for~$R_{2,j}$ instead of $R^{2}_{1,j}$. From the SDE for~$R_{2,j}$,
\begin{align}
	\EE\left( R_{2,j, t \wedge T_{M}} \right) 
	& = R_{2,j,0} + \lambda_{2,j} \int_{0}^{t}  \EE \left(( \sigma^{2}_{s} - R_{2,j,s}) \1_{s \leq t \wedge T_{M}} \right) ds. \label{4factorr2eq} 
\end{align}
Focusing again on $j = 0$ first, we insert the definitions of $\sigma^{2}_{s}$ and $R_{2,s}$ and estimate
\begin{align*}
	\sigma^{2}_{s} & - R_{2,0,s} \\
	& = (\beta_{0} + \beta_{1}R_{1,s})^{2} + 2\beta_{2}(\beta_{0} + \beta_{1}R_{1,s})\sqrt{R_{2,s}} + \beta^{2}_{2}\theta_{2}R_{2,1,s} + (\beta^{2}_{2}(1-\theta_{2}) - 1)R_{2,0,s} \\
	& \leq 2\beta_{0}^{2} + 2\beta_{1}^{2}R_{1,s}^{2} + \beta_{0}^{2} + \beta_{2}^{2}R_{2,s} + \beta_{1}^{2}R_{1,s}^{2} + \beta_{2}^{2}R_{2,s} + \beta^{2}_{2}\theta_{2}R_{2,1,s} + (\beta^{2}_{2}(1-\theta_{2}) - 1)R_{2,0,s} \\
	& \leq 3\beta_{0}^{2} + 3\beta_{1}^{2}R_{1,s}^{2} + 3\beta_{2}^{2}R_{2,1,s} + (3\beta_{2}^{2} - 1)R_{2,0,s}\\
        & \leq 3\beta_{0}^{2} + 3\beta_{1}^{2}R_{1,0,s}^{2} + 3\beta_{1}^{2}R_{1,1,s}^{2} + 3\beta_{2}^{2}R_{2,1,s} + (3\beta_{2}^{2} - 1)R_{2,0,s}\,.
\end{align*}
Applying this in \eqref{4factorr2eq}, we conclude that
\begin{align*}
	\EE\left( R_{2,0, t \wedge T_{M}} \right) & \leq R_{2,0,0} + \lambda_{2,0} \int_{0}^{t}  \EE \Big( 3\beta_{0}^{2} + 3\beta_{1}^{2}R_{1,0,s \wedge T_{M}}^{2} \\
 & \quad\,+ 3\beta_{1}^{2}R_{1,1,s \wedge T_{M}}^{2} + 3\beta_{2}^{2}R_{2,1,s \wedge T_{M}} + (3\beta_{2}^{2} - 1)R_{2,0,s \wedge T_{M}}\Big) ds \\
	& =: c_{2,0,0}(t) + c_{2,0,1}\int_{0}^{t} \EE\left(R^{2}_{1,0,s \wedge T_{M}} \right) ds + c_{2,0,1}\int_{0}^{t} \EE\left(R^{2}_{1,1,s \wedge T_{M}} \right) ds \\
 & \quad\,+  c_{2,0,1}\int_{0}^{t} \EE\left(R_{2,1,s \wedge T_{M}} \right) ds + c_{2,0,2}\int_{0}^{t} \EE\left(R_{2,0,s \wedge T_{M}} \right) ds.
\end{align*}
Symmetrically, we obtain for $j = 1$ that
\begin{align*}
	\sigma^{2}_{s} - R_{2,1,s} 
	& \leq 3\beta_{0}^{2} + 3\beta_{1}^{2}R_{1,0,s}^{2} + 3\beta_{1}^{2}R_{1,1,s}^{2} + 3\beta_{2}^{2}R_{2,0,s} + (3\beta_{2}^{2} - 1)R_{2,1,s} 
\end{align*}
and then
\begin{align*}
	\EE\left( R_{2,1, t \wedge T_{M}} \right) & \leq R_{2,1,0} + \lambda_{2,0} \int_{0}^{t}  \EE \Big( 3\beta_{0}^{2} + 3\beta_{1}^{2}R_{1,0,s}^{2}  \\
 & \quad\,+ 3\beta_{1}^{2}R_{1,1,s}^{2} + 3\beta_{2}^{2}R_{2,0,s} + (3\beta_{2}^{2} - 1)R_{2,1,s} \Big) ds \\
	&  =: c_{2,1,0}(t) + c_{2,1,1}\int_{0}^{t} \EE\left(R^{2}_{1,0,s \wedge T_{M}} \right) ds + c_{2,1,1}\int_{0}^{t} \EE\left(R^{2}_{1,1,s \wedge T_{M}} \right) ds \\
 & \quad\,+ c_{2,1,1}\int_{0}^{t} \EE\left(R_{2,0,s \wedge T_{M}} \right) ds + c_{2,1,2}\int_{0}^{t} \EE\left(R_{2,1,s \wedge T_{M}} \right) ds.
\end{align*}
Adding the two inequalities, we deduce
\begin{align}
	\EE & \left(  R_{2,0, t \wedge T_{M}} 
	 + R_{2,1, t \wedge T_{M}} \right) \nonumber \\
	& \leq c_{2,0}(t) + c_{2,1}\int_{0}^{t} \EE(R^{2}_{1,0, s \wedge T_{M}} + R^{2}_{1,1, s \wedge T_{M}})ds \nonumber \\
 & \quad\,+ c_{2,2} \int_{0}^{t} \EE\left(  R_{2,0, s \wedge T_{M}} + R_{2,1, s \wedge T_{M}} \right)ds. \label{4factorr2sumeq}
\end{align}
We can now add \eqref{4factorr1sumeq} and \eqref{4factorr2sumeq} to obtain
\begin{align*}
	\EE(U_{t \wedge T_{M}} ) \leq c_{0}(t) + c_{1} \int_{0}^{t} \EE(U_{s \wedge T_{M}})ds 
\end{align*}
where $c_{0}(t)$ is affine and nondecreasing in~$t$ and $c_{0}(t),c_{1}$ do not depend on~$M$. Gr\"onwall's inequality then yields
\begin{align*}
	\EE(U_{t \wedge T_{M}} ) \leq c_{0}(t) e^{c_{1} t}
\end{align*}
which is the desired bound~\eqref{eq:4factornonexplosionProof}.
\end{proof}

\begin{remark}\label{4factorgen} The results in this section generalize to volatility models having the same dynamics as~\eqref{4factor} for $(R_{1,0,t},R_{1,1,t},R_{2,0,t}, R_{2,1,t})$ but a more general functional \[\sigma_{t} = \tilde{f}(R_{1,t},R_{2,t}) = f(R_{1,0,t},R_{1,1,t},R_{2,0,t},R_{2,1,t}),\] where 
\begin{enumerate}
  \item $f: \RR^{2} \times \RR^{2}_{+} \to \RR$ is continuous, and Lipschitz on compact subsets of $\RR^{2} \times (0,\infty)^{2}$,
  \item there are $K_{1}, K_{2} \in \RR$ such that
  \begin{align}
            f(x_{0},x_{1},y_{0},y_{1})^{2} \leq K_{1} + K_{2} (x_{0}^{2} + x_{1}^{2} + y_{0} + y_{1}).\label{4factorgenbound}
  \end{align}   
\end{enumerate}
These conditions for~$f$ are satisfied in particular if $\tilde{f}$ satisfies the conditions of \cref{4factorgen}; this can be seen using inequalities~\eqref{r1convexineq} and~\eqref{r2convexineq}.

Indeed, continuity and local Lipschitz continuity imply the analogue of \cref{pr:wellposed2}. Inequality \eqref{4factorgenbound} implies that bounding~\eqref{eq:4factornonexplosionProof} is enough to bound the expectation of $\sigma^{2}_{t}$. Moreover, it implies that, for $j \in \left\{ 0, 1\right\}$, $f$ satisfies
    \begin{align}
    \lambda_{1,j}^{2}f(x_{0},x_{1},y_{0},y_{1})^{2} - 2\lambda_{1,j}x_{j}^{2} \leq K_{1,j,0} + K_{1,j,1}(x_{0}^{2} + x_{1}^{2} + y_{0} + y_{1}) \label{4factorgenboundr1}
    \end{align}
    for some $K_{1,j,0}, K_{1,j,1} \in \RR$, so that a bound analogous to~\eqref{4factorr1sumeq}  holds, as well as, for $j \in \left\{ 0, 1\right\}$,
        \begin{align*}
           f(x_{0},x_{1},y_{0},y_{1})^{2} - y_{j} \leq K_{1} + K_{2}(x_{0}^{2} + x_{1}^{2} + y_{0} + y_{1}), %
        \end{align*}
        so that a bound analogous to~\eqref{4factorr2sumeq} holds.
\end{remark}

\subsection{Failure of Positivity of $\sigma_{t}$ in~\eqref{4factor}}\label{se:4factorpos}

In \cite{GuyonLekeufack.22} it is reported that for realistic parameter values, the volatility $\sigma_{t}$ remained positive in simulations of the 4-factor model. Nevertheless, existence of a reasonable sufficient condition for $\sigma_{t}>0$ (or even just $\sigma_{t}\geq0$) in~\eqref{4factor} remains open. Below, we explain that a direct generalization of \cref{2sigmapos} fails. Indeed, in the 4-factor model, $\sigma_{t}$ follows the SDE
\begin{align}
	d\sigma_{t} = \left(-\beta_{1}\bar{\lambda}_{1}\bar{R}_{1,t} + \frac{\bar{\lambda}_{2}\beta_{2}}{2} \frac{\sigma^{2}_{t} - \bar{R}_{2,t}}{\sqrt{R_{2,t}}} \right) dt + \beta_{1}\bar{\lambda}_{1}\sigma_{t}dW_{t}\label{4factorsigmaeq}
\end{align}
where
\begin{align*}
	\bar{\lambda}_{i} := (1 - \theta_{i})\lambda_{i,0} + \theta_{i}\lambda_{i,1}\,,\qquad \bar{R}_{i,t} :=	\frac{(1 - \theta_{i})\lambda_{i,0}R_{i,0,t} + \theta_{i}\lambda_{i,1}R_{i,1,t}}{\bar{\lambda}_{i}}
\end{align*}
as seen in \cite{GuyonLekeufack.22}. The analogue of the condition in \cref{2sigmapos} is $\bar{\lambda}_{2} < 2 \bar{\lambda}_{1}$.

\begin{proposition}\label{pr:4sigmaposFail}
   Under~\eqref{4factor}, it may happen that $\sigma_{0}>0$ but $\PP(\sigma_{t}<0)>0$ for some $t>0$, even if $\bar{\lambda}_{2} < 2 \bar{\lambda}_{1}$.
\end{proposition} 

\begin{proof}
  We choose initial conditions $R_{1,0,0}<0$ and $R_{1,1,0}>0$, and coefficients $\theta_{1},\lambda_{1,0},\lambda_{1,1}$, such that $R_{1,0}> 0$ and $\bar{R}_{1,0}<0$. Next, choose $\beta_{1}=-1$ (say), and then $\beta_{2}>0$ such that $\beta_{1}R_{1,0} + \beta_{2} \sqrt{R_{2,0}}=0$. Consider for the moment $\beta_{0}:=0$, then the preceding identity means that $\sigma_{0}=0$. 
  Inspecting~\eqref{4factorsigmaeq}, we see that at $t=0$, the volatility vanishes while the drift rate is
\begin{align*}
	-\beta_{1}\bar{\lambda}_{1}\bar{R}_{1,t} - \frac{\bar{\lambda}_{2}\beta_{2}}{2} \frac{\bar{R}_{2,t}}{\sqrt{R_{2,t}}} <0.
\end{align*}  
  By continuity of the paths, it follows that $\PP(\sigma_{t}<0)>0$ for all $t>0$ sufficiently small.
  
  Next, we modify the above by choosing $\beta_{0}$ strictly positive, so that $\sigma_{0}=\beta_{0}>0$. We may see $\beta_{0}$ as a parameter of the SDE determining $\sigma_{t}$. If the solution is continuous with respect to $\beta_{0}$, it follows that $\PP(\sigma_{t}<0)>0$ for $\beta_{0}>0$ and $t$ sufficiently small. Continuity is a standard result for SDEs with Lipschitz coefficients (e.g., \cite[Section~4.5]{Kunita.90}). To see that the Lipschitz result is sufficient, note that for the present purpose of showing that $\PP(\sigma_{t}<0)>0$ for some small $t>0$, we may truncate the non-Lipschitz coefficients in~\eqref{4factor}; that is, we replace $\sqrt{R_{2,t}}$ by $\sqrt{ R_{2,t} \vee \delta}$ and $\sigma_{t}^{2}$ by $\sigma_{t}^{2}\wedge \delta^{-1}$ for a small constant $\delta>0$. 
\end{proof}

\section{Martingale Property of the Price Process}\label{se:martprop}

In this section we study the martingale property of the price process $(X_{t})_{t\geq 0}$; cf.\ \cref{martprop}. Following~\cite[Proof of Lemma 4.2]{Sin.98} and \cite[Section 5.3]{JarrowProtterSanMartin.22}, the idea is to rephrase the martingale property into non-explosiveness of another process, and establish the latter. 
Given a continuous adapted process $(\nu_{t})_{t\geq0}$, consider the exponential local martingale $(X_{t})_{t\geq 0}$ given by 
\begin{align}
    dX_{t} = \nu_{t}X_{t}dW_{t}, \quad X_{0} = x_{0} > 0.\label{priceprocess}
\end{align}
For $M > 0$, we define the stopping time
\begin{align}
    T_{M} := \inf \left\{ t \geq 0 \text{ : } |\nu_{t}| \geq M \right\}. \label{sigmatime}
\end{align}
Given also $t\geq0$, Novikov's condition allows us to define a probability measure $\tilde{\PP}_{M,t}\ll \PP$ by
\begin{align*}
    \frac{d\tilde{\PP}_{M,t}} {d\PP} = \exp \left\{ \int_{0}^{T_{M} \wedge t} \nu_{s} dW_{s} - \frac{1}{2} \int_{0}^{T_{M} \wedge t} \nu^{2}_{s} ds \right\}.
\end{align*}
We have $T_{M}\to\infty$ a.s.\ given that $(\nu_{t})_{t\geq0}$ does not explode in finite time. The martingale property of $(X_{t})_{t\geq 0}$ is characterized as follows.

\begin{lemma}\label{le:trueMartChar}
  For the local martingale  $(X_{t})_{t\geq 0}$ given by~\eqref{priceprocess}, the following are equivalent:
  \begin{enumerate}
  \item $(X_{t})_{t\geq 0}$ is a martingale,
  \item $\liminf_{M \to \infty}\EE(X_{T_{M}}\1_{T_{M} < t}) = 0$ for all $t > 0$,
  \item $\liminf_{M \to \infty} \tilde{\PP}_{M,t}( T_{M} < t) = 0$ for all $t > 0$.
  \end{enumerate}
\end{lemma} 

\begin{proof}
As $(X_{t})_{t\geq 0}$ is a nonnegative local martingale, it is a supermartingale by Fatou's lemma, hence a martingale if and only if $x_{0} = \EE(X_{t})$ for all $t>0$. The bounded process $(X_{t \wedge T_{M}})_{t\geq 0}$ is a martingale, hence
\begin{align*}
    x_{0} = \EE(X_{t \wedge T_{M}}) = \EE(X_{t}\1_{T_{M} \geq t}) + \EE(X_{T_{M}}\1_{T_{M} < t}).
\end{align*}
Using monotone convergence for $\EE(X_{t}\1_{T_{M} \geq t})$, this yields
\begin{align*}
    x_{0} = \EE(X_{t}) + \liminf_{M \to \infty}\EE(X_{T_{M}}\1_{T_{M} < t})
\end{align*}
and now the equivalence of (i) and (ii) follows. For fixed $M > 0$ and $t \geq 0$,
\begin{align}
    \EE(X_{T_{M}}\1_{T_{M} < t}) & = \EE\left( x_{0} \exp \left\{ \int_{0}^{T_{M}} \nu_{s} dW_{s} - \frac{1}{2} \int_{0}^{T_{M}} \nu^{2}_{s} ds \right\} \1_{T_{M} < t} \right) \nonumber \\
    & =  x_{0}  \EE\left(\exp \left\{ \int_{0}^{T_{M} \wedge t} \nu_{s} dW_{s} - \frac{1}{2} \int_{0}^{T_{M} \wedge t} \nu^{2}_{s} ds \right\} \1_{T_{M} < t} \right) \nonumber \\
    & = x_{0}\tilde{\PP}_{M,t}( T_{M} < t),\label{measurechange}
\end{align}
showing the equivalence of (ii) and (iii).
\end{proof} 

\begin{proposition}\label{pr:trueMart2}
   Let $(\sigma_{t})_{t\geq0}$ be given by~\eqref{2factor}. For a fixed constant $C\geq 0$, define the stopping time
$
    \tau := \inf \left\{ t \geq 0:\sigma_{t} < -C \right\} 
$
and set $\nu_{t}:=\sigma_{t\wedge \tau}$. Then the exponential local martingale $(X_{t})_{t\geq 0}$ of~\eqref{priceprocess} is a true martingale.
\end{proposition} 

\begin{proof}
Define the stopping time
\begin{align*}
    S_{M} := \inf \left\{ t \geq 0 \text{ : } |\sigma_{t}| \geq M \right\}. %
\end{align*}
Fix $t, M>0$. By Girsanov's theorem, 
\begin{align*}
    \tilde{W}_{s} = W_{s} - \int_{0}^{s}\sigma_{r}dr \text{ , } \quad 0 \leq s < t \wedge S_{M}
\end{align*}
is a Brownian motion under $\tilde{\PP}_{M,t}$ up to time $t \wedge S_{M}$. The system~\eqref{2factor} for $(R_{1,s},R_{2,s})$ can be stated up to time $t \wedge S_{M}$ as
\begin{align*}
    \sigma_{s} & = \beta_{0} + \beta_{1}R_{1,s} + \beta_{2}\sqrt{R_{2,s}}  \nonumber \\
	dR_{1,s} & = \lambda_{1}\sigma_{s}d\tilde{W}_{s} + \lambda_{1}(\sigma^{2}_{s} - R_{1,s})ds  \\
	dR_{2,s} & = \left(\lambda_{2}\sigma_{s}^{2} - \lambda_{2}R_{2,s} \right) ds.  \nonumber
\end{align*}
By the same arguments as in \cref{weakexist} and \cref{pathwiseUnique}, this system has a unique strong solution up to time $t \wedge S_{M}$ under $\tilde{\PP}_{M,t}$. The solution $(\sigma_{s})$ of~\eqref{2factor} up to time $t \wedge S_{M}$ under $\tilde{\PP}_{M,t}$ has the same distribution as the solution of~\eqref{2factortilde} in \cref{le:tiltedPDV2} below under~$\PP$. For $M>C$, we also note that $t \wedge S_{M}\geq t \wedge T_{M}\wedge \tau$. 
By the assertion of \cref{le:tiltedPDV2}, it then follows that $\lim_{M\to\infty}\tilde{\PP}_{M,t}(T_{M} < t) = 0$ for any $t\geq0$, which by \cref{le:trueMartChar} shows that $(X_{t})_{t\geq 0}$ is a martingale.
\end{proof}

\begin{lemma}\label{le:tiltedPDV2}
  The following SDE under $\PP$ has a unique strong solution $(\sigma_{s})_{s\geq0}$ up to a possible time of explosion:
\begin{align}
    \sigma_{s} & = \beta_{0} + \beta_{1}R_{1,s} + \beta_{2}\sqrt{R_{2,s}}  \nonumber \\
	dR_{1,s} & = \lambda_{1}\sigma_{s}dW_{s} + \lambda_{1}(\sigma^{2}_{s} - R_{1,s})ds  \label{2factortilde} \tag{\mbox{$\text{2-PDV}{\tilde{\phantom{e}}}$}} \\
	dR_{2,s} & = \left(\lambda_{2}\sigma_{s}^{2} - \lambda_{2}R_{2,s} \right) ds.  \nonumber
\end{align}  
   Given $C\geq0$, define $\tau=\inf\{t\geq0: \sigma_{s}<-C\}$. Then the process $(\sigma_{s\wedge \tau})_{s\geq0}$ a.s.\ does not explode in finite time.
\end{lemma} 

\begin{proof}
  Existence and uniqueness again follows as in \cref{weakexist} and \cref{pathwiseUnique}. For $M>C$, let
  \begin{align*}
    T_{M}&:=\inf\{s\geq 0: |\sigma_{s\wedge \tau}|\geq M\}=\inf\{s\geq 0: \sigma_{s}\geq M\},\\
    S_{M}&:=\inf\{s\geq 0: |\sigma_{s}|\geq M\}.
  \end{align*}
  Note that $\tau\leq T_{M}$ implies $T_{M}=\infty$. For fixed $t\geq0$, we can then write
    \begin{align}
        M\PP(T_{M} \leq t) 
        &= \EE(\sigma_{t \wedge T_{M}}\1_{T_{M} \leq t}) = \EE(\sigma_{t \wedge T_{M}}\1_{T_{M} \leq t}\1_{\tau > T_{M}}) \nonumber\\
        &= \EE(\sigma_{t \wedge S_{M}}\1_{T_{M} \leq t}\1_{\tau > T_{M}}) 
        \leq \EE((C+\sigma_{t \wedge S_{M}})\1_{T_{M} \leq t}\1_{\tau > T_{M}}) \nonumber\\
        &\leq \EE(C+\sigma_{t \wedge S_{M}}) = C+ \EE(\sigma_{t \wedge S_{M}}).
        \label{2factortildeprobineq}
    \end{align}
    Below, we show a bound $\EE(\sigma_{t \wedge S_{M}})\leq K_{0} + K_{1}t$ that is uniform in~$M$. Then, \eqref{2factortildeprobineq} implies that $\lim_{M\to\infty} \PP(T_{M} \leq t)=0$, which is the claim.
    
    Choose $\hat{\beta}_{2} > 0$ small enough such that $\beta_{1}\lambda_{1} + \hat{\beta}_{2}\lambda_{2} < 0$, and then choose $\bar{\beta}_{2} > 0$ such that $\beta_{2} \sqrt{x} \leq \bar{\beta}_{2} + \hat{\beta}_{2}x$ for all $x \geq 0$. The first equation in~\eqref{2factortilde} then yields
    \begin{align}
        \EE(\sigma_{t \wedge S_{M}}) \leq \beta_{0} + \bar{\beta}_{2} + \EE(\beta_{1}R_{1, t \wedge S_{M} } + \hat{\beta}_{2} R_{2, t \wedge S_{M}}) \label{2factortildesigmabound}
    \end{align}
    and we can focus on bounding the last expectation. Using the last two equations in~\eqref{2factortilde}, we have 
    \begin{align}
        \EE(\beta_{1}R_{1, t \wedge S_{M} } + \hat{\beta}_{2} R_{2, t \wedge S_{M}}) = \beta_{1}R_{1,0 } + \hat{\beta}_{2} R_{2,0} + & \EE \Big( \int_{0}^{\infty} \Big\{(\beta_{1}\lambda_{1} + \hat{\beta}_{2}\lambda_{2})\sigma^{2}_{s} \nonumber \\ 
          & - \lambda_{1}\beta_{1}R_{1,s} - \lambda_{2}\hat{
        \beta}_{2}R_{2,s} \Big\} \1_{s \leq t \wedge S_{M}} ds \Big). \label{eq:2factorItoExpBoundtilde}
    \end{align}
Defining $\alpha := \beta_{1}\lambda_{1} + \hat{\beta}_{2}\lambda_{2} < 0$, the term under the integral is
    \begin{align*}
        \alpha\sigma^{2}_{s} - \lambda_{1}\beta_{1}R_{1,s} - \lambda_{2}\hat{
        \beta}_{2}R_{2,s} 
        & = \alpha\beta^{2}_{0} + \alpha \beta^{2}_{2}R_{2,s} + \alpha\beta_{1}^{2}R_{1,s}^{2} + 2\alpha\beta_{0}\beta_{1}R_{1,s} \\
        & \quad\, + 2\alpha\beta_{0}\beta_{2}\sqrt{R_{2,s}} + 2\alpha\beta_{1}\beta_{2}R_{1,s}\sqrt{R_{2,s}} - \lambda_{1}\beta_{1}R_{1,s} - \lambda_{2}\hat{
        \beta}_{2}R_{2,s} \\
        & = C + B\sqrt{R_{2,s}} + AR_{2,s}
    \end{align*}
    where
    \begin{align*}
        C & = \alpha\beta^{2}_{0} + \alpha\beta_{1}^{2}R_{1,s}^{2} + 2\alpha\beta_{0}\beta_{1}R_{1,s} - \lambda_{1}\beta_{1}R_{1,s} \text{ ,} \\
        B & = 2\alpha\beta_{2}(\beta_{0} + \beta_{1}R_{1,s}) \text{ , } \quad A = \alpha\beta_{2}^{2} - \lambda_{2}\hat{\beta}_{2}
    \end{align*}
  and $A<0$ due to $\alpha<0$ and $\lambda_{2}\hat{\beta}_{2} > 0$. Thus, the term is bounded from above by 
    \begin{align*}
        C - \frac{B^{2}}{4A} & = \alpha\beta^{2}_{0} + \alpha\beta_{1}^{2}R_{1,s}^{2} + 2\alpha\beta_{0}\beta_{1}R_{1,s} - \lambda_{1}\beta_{1}R_{1,s} - \frac{4\alpha^{2}\beta^{2}_{2}(\beta_{0}^{2} + \beta_{2}^{2}R_{1,s}^{2} + 2\beta_{0}\beta_{1}R_{1,s})}{4(\alpha\beta_{2}^{2} - \lambda_{2}\hat{\beta}_{2})} \\
        & = C' + B'R_{1,s} + A'R_{1,s}^{2}
    \end{align*}
    with
    \begin{align*}
        C' = \alpha\beta^{2}_{0} + \frac{\alpha^{2}\beta^{2}_{0}\beta^{2}_{2}}{\lambda_{2}\hat{\beta}_{2} - \alpha\beta_{2}^{2}}, &  \quad B' = 2\alpha\beta_{0}\beta_{1} - \lambda_{1}\beta_{1} + \frac{2\alpha^{2}\beta^{2}_{2}\beta_{0}\beta_{1}}{\lambda_{2}\hat{\beta}_{2} -\alpha\beta_{2}^{2}},\\
        A' & = \alpha\beta_{1}^{2} + \frac{\alpha^{2}\beta_{1}^{2}\beta_{2}^{2}}{\lambda_{2}\hat{\beta}_{2} -\alpha\beta_{2}^{2}}.
    \end{align*}
    Here again $A' < 0$ due to $\alpha<0$ and $\lambda_{2}\hat{\beta}_{2} > 0$,
so that the term is bounded from above by the constant $L:=C' - \frac{B'^{2}}{4A'}$. Using this bound in~\eqref{eq:2factorItoExpBoundtilde} yields
    \begin{align*}
        \EE(\beta_{1}R_{1, t \wedge S_{M} } + \hat{\beta}_{2} R_{2, t \wedge S_{M}}) & \leq \beta_{1}R_{1,0 } + \hat{\beta}_{2} R_{2,0} + \EE\left( \int_{0}^{\infty} L \1_{s \leq t \wedge S_{M}} ds \right) \\
        & \leq \beta_{1}R_{1,0 } + \hat{\beta}_{2} R_{2,0} + |L|t
    \end{align*}
    and thus
    \begin{align*}
        \EE(\sigma_{t \wedge S_{M}}) \leq \beta_{0} + \bar{\beta}_{2} + \beta_{1}R_{1,0 } + \hat{\beta}_{2}R_{2,0} + |L|t =: K_{0} + K_{1}t
    \end{align*}
    as claimed.
\end{proof}

\begin{remark}\label{2factortildegen} 
  \Cref{pr:trueMart2} and its proof generalize to volatility models having the same dynamics as~\eqref{2factor} for $(R_{1,t},R_{2,t})$ but a more general functional $\sigma_{t} = f(R_{1,t},R_{2,t})$, where $f$ satisfies the conditions of \cref{2factorgen} and in addition there exist constants $L_{1},L_{2},L_{3},L\in\RR$ such that
\begin{gather*}
    f(x,y) \leq L_{0} + L_{1}x + L_{2}y, \\
    (L_{1}\lambda_{1} + L_{2}\lambda_{2})f(x,y)^{2} - \lambda_{1}L_{1}x - \lambda_{2}L_{2}y \leq L.
\end{gather*}
The constant $L_{1}$ typically needs to be negative, as in~\eqref{2factor} where $L_{1}=\beta_{1}<0$.
\end{remark}

Our final remark details why the above proof does not extend to the 4-factor model. 

\begin{remark}\label{rk:tiltedPDV4}
  The general line of argument given above may extend to~\eqref{4factor}. Indeed, the following SDE under $\PP$ has a unique strong solution $(\sigma_{s})_{s\geq0}$ up to a possible time of explosion:
\begin{align}
	\sigma_{t} & = \beta_{0} + \beta_{1}R_{1,t} + \beta_{2}\sqrt{R_{2,t}} \nonumber \\
	R_{1,t} & = (1 - \theta_{1})R_{1,0,t} + \theta_{1}R_{1,1,t} \nonumber \\
	R_{2,t} & = (1 - \theta_{2})R_{2,0,t} + \theta_{2}R_{2,1,t}\label{4factortilde}\tag{\mbox{$\text{4-PDV}{\tilde{\phantom{e}}}$}} \\
	dR_{1,j,t} & = \lambda_{1,j}\sigma_{t}dW_{t} +  \lambda_{1,j}( \sigma^{2}_{t} - R_{1,j,t})dt, \quad j \in \{0,1\} \nonumber \\
	dR_{2,j,t} & = \lambda_{2,j}\left(\sigma_{t}^{2} - R_{2,j,t} \right) dt, \quad j \in \{0,1\}.\nonumber 
\end{align}
   Given $C\geq0$, define $\tau=\inf\{t\geq0: \sigma_{s}<-C\}$. If the process $(\sigma_{s\wedge \tau})_{s\geq0}$ a.s.\ does not explode in finite time, then the assertion of \cref{martprop} extends to~\eqref{4factor}. However, the proof for non-explosiveness given in \cref{le:tiltedPDV2} does not extend directly to the present setting.
   
   Indeed, the system~\eqref{4factortilde} satisfies bounds analogous to~\eqref{2factortildeprobineq} and~\eqref{2factortildesigmabound}. Applying the same procedure as in the proof of  \cref{le:tiltedPDV2}, and recalling the notation used in~\eqref{4factorsigmaeq}, 
   \begin{align*}
       \EE(\beta_{1}R_{1, t \wedge S_{M} } + \hat{\beta}_{2} R_{2, t \wedge S_{M}}) = & \beta_{1}R_{1,0 } + \hat{\beta}_{2} R_{2,0} + \EE \Big( \int_{0}^{\infty} \Big\{(\beta_{1}\bar{\lambda}_{1} + \hat{\beta}_{2}\bar{\lambda}_{2})\sigma^{2}_{s} \nonumber \\ 
        & - \lambda_{1,1}\theta_{1}\beta_{1}R_{1,1,s} -\lambda_{1,0}(1 - \theta_{1})\beta_{1}R_{1,0,s} \nonumber \\
        & - \lambda_{2,1} \theta_{2}\hat{
        \beta}_{2}R_{2,1,s} - \lambda_{2,0}(1 - \theta_{2})\hat{
        \beta}_{2}R_{2,0,s} \Big\} \1_{s \leq t \wedge S_{M}} ds \Big).
   \end{align*}
   The integrand is bounded from above by a quadratic form in $(R_{1,0,s}, R_{1,1,s}, \sqrt{R_{2,0,s}}, \sqrt{R_{2,1,s}})$; however, the matrix defining the form is not negative definite. Hence, we cannot  bound it uniformly as we did in \cref{le:tiltedPDV2}.
\end{remark}

 \newcommand{\dummy}[1]{}

\end{document}